%% file: main.tex
\documentclass{lmcs} 
\pdfoutput=1

\usepackage{lastpage}
\lmcsdoi{15}{4}{14}
\lmcsheading{}{\pageref{LastPage}}{}{}%
{Feb.~27,~2018}{Dec.~11,~2019}{}


\usepackage{hyperref}
\usepackage[all]{xy}
\xyoption{arc}
\usepackage{todonotes}
\input{macros}

\begin{document}

\title{Definable isomorphism problem}

\author[K.~Keshvardoost]{Khadijeh Keshvardoost\rsuper{a}}	
\address{\lsuper{a}Department of Mathematics, Velayat University, Iranshahr, Iran}	

\author[B.~Klin]{Bartek Klin\rsuper{b}}	
\address{\lsuper{b}University of Warsaw}	

\author[S.~Lasota]{S{\l}awomir Lasota\rsuper{b}}	
\thanks{
This work has been partially supported by the ERC project `Lipa' within the
EU Horizon 2020 research and innovation programme, No. 683080
(the second author), the NCN grant 2016/21/B/ST6/01505 (the third author) and
the NCN grant 2016/21/D/ST6/01485 (the fifth author).
}

\address{\lsuper{c}CNRS, LaBRI, Universit\'e de Bordeaux}	
\author[J.~Ochremiak]{Joanna Ochremiak\rsuper{c}}	

\author[Sz.~Toru{\'n}czyk]{Szymon Toru{\'n}czyk\rsuper{b}}	

\begin{abstract}
We investigate the isomorphism problem in the setting of definable sets (equivalent to sets with atoms):
given two definable relational structures, are they related by a definable isomorphism?
Under mild assumptions on the underlying structure of atoms, we prove decidability of the problem. 
The core result is parameter-elimination: existence of
an isomorphism definable with parameters implies existence of an isomorphism definable without parameters.
\end{abstract}

\maketitle


\input{intro}

\input{prelim}

\input{defiso}

\input{orbitfinite}

\input{elim}

\bibliographystyle{alpha}
\bibliography{bib}

\end{document}

%% file: macros.tex
\newcommand{\para}[1]{\mypar{#1}}  

\def\0{$\emptyset$}

\theoremstyle{definition}
\newtheorem{example}[thm]{Example} 
\theoremstyle{plain}
\newtheorem{remark}[thm]{Remark} 
\newtheorem{theorem}[thm]{Theorem} 
\newtheorem{lemma}[thm]{Lemma} 
\newtheorem{corollary}[thm]{Corollary} 
\newtheorem{claim}[thm]{Claim} 
\newcommand{\mypar}[1]{\subsection*{#1}}

\newcommand*\from{\colon}
\newcommand{\st}{\ |\ }

\newcommand{\setof}[2]{\left\{\, #1 \, | \, #2 \, \right\}}

\newcounter{quotecount}

\renewcommand{\subset}{\subseteq}

\newcommand{\atoms}{\mathsf{Atoms}} 
\newcommand{\Atoms}{\atoms}   
\newcommand{\eqdef}{\stackrel {\text{def}} =}

\renewcommand{\setminus}{-}

\newcommand{\tri}[4]{
  \vcenter{\xy \xygraph{!{/r9px/:} 
       [] [rrr]{\scriptscriptstyle{\circlearrowright}} 
          !P3{~={#4}~><{@{.}}~*{\ifcase\xypolynode \or #2 \or #1 \or #3\fi}}}
       \endxy}
}

\newcommand{\fixtri}[4]{
  \vcenter{\xy \xygraph{!{/r9px/:} 
       []  !P3{~={#4}~><{@{.}}~*{\ifcase\xypolynode \or #2 \or #1 \or #3\fi}}}
  \endxy}
}

\newcommand{\dom}{\text{Dom}}
\renewcommand{\int}[1]{{[#1]}}

\newcommand{\arc}[2]{\textit{cons}_{k,l}}

\newcommand{\set}[1]{\{#1\}}

\newcommand{\str}[1]{{\mathbb{#1}}}

\newcommand{\problem}[1]{{\sc #1}}
\newcommand{\decproblemdef}[3]
{\medskip\noindent
{\sffamily Problem:} {\problem{#1}}\\
{\sffamily Input:} #2\\
{\sffamily Decide:} #3
}

\newcommand{\atom}[1]{\underline{#1\mkern-1mu}\mkern1mu}

\newcommand{\val}{\mathrm{val}}

\newcommand{\supp}[1]{\mathrm{sup}(#1)}
\newcommand{\supparg}[2]{\mathrm{supp}_{#1}(#2)}

\newcommand{\codom}{\text{Codom}}

\newcounter{col}

{
 \stepcounter{col}
  \collect{col}%
  {\subsection{Proof of #1}%
  \expandafter\label\expandafter{col:\arabic{col}}%
  \stepcounter{col}%
  }
  {\qed
  }
}{\endcollect
  \smallskip}

{
		\stepcounter{col}\collect{col}%
  {\subsection{Proof of #1}%
  \expandafter\label\expandafter{col:\arabic{col}}%
  \stepcounter{col}%
  }%
  {\qed}%
}{\endcollect}

{
 \stepcounter{col}\collect{col2}%
  {\subsection{Proof of #1}%
  \expandafter\label\expandafter{col:\arabic{col}}%
  \stepcounter{col}%
  }
  {\qed
  }
}{\endcollect
  \smallskip}

\newcommand{\colsec}[1]{}

\newif\ifisfinalversion
\isfinalversionfalse

\theoremstyle{definition}

\newcommand{\length}[1]{\text{length}(#1)}
\newcommand{\order}[1]{\text{order}(#1)}
\newcommand{\orbit}[2]{\text{orbit}_{#1}(#2)}


%% file: intro.tex

\newcommand{\spl}{\raisebox{1pt}{$\scriptscriptstyle+$}}
\medskip

\section{Introduction}
We consider \emph{hereditarily first-order definable sets}, which are usually infinite, but can be finitely described and are therefore
amenable to algorithmic manipulation. 
We drop the qualifiers \emph{herediatarily first-order}, and simply call them definable sets in what follows. They are parametrized by a fixed underlying relational structure
$\atoms$ whose elements are called \emph{atoms}.
%
%
\begin{example}\label{ex:intro}
Let $\atoms$ be a countable set $\{\atom{1},\atom{2},\atom{3},\ldots\}$ equipped with
the equality relation only; we shall call this structure the {\em pure set}.
  Let 
\begin{align*}
V&=\setof{\set{a,b}}{a,b\in\atoms, a\neq b},\\
E&=\setof{(\set{a,b},\set{c,d})}{a,b,c,d\in\atoms
, a\neq b \wedge a \neq c \wedge a\neq d \wedge b \neq c \wedge b \neq d \wedge c \neq d}.
\end{align*}
Both $V$ and $E$ are  definable sets
(over $\atoms$),
as they are constructed from $\atoms$ using (possibly nested) set-builder expressions with first-order guards ranging over $\atoms$. In general, we allow finite unions in the definitions, and 
finite tuples (as above) are allowed for notational
convenience.
Precise definitions are given in Section~\ref{sec:atoms}. The pair $G=(V,E)$
is also a definable set, in fact, a definable graph. It is an infinite Kneser graph 
(a generalization of the famous Petersen graph):
its vertices are all two-element subsets of $\atoms$,
and two such subsets are adjacent iff they are disjoint.

The graph $G$ is {\em $\emptyset$-definable}: its definition does not refer to any particular 
elements of $\atoms$. In general, 
one may refer to a finite set of parameters  $S\subset\atoms$ to describe an \emph{$S$-definable} set.
 For instance, the set $\setof{a}{a\in\atoms,a\neq \atom 1\land a\neq \atom 2}$ is $\set{\atom 1,\atom 2}$-definable. 
 Definable sets are those which are $S$-definable for some finite $S\subset\atoms$.
 \qed
\end{example}

We remark that in the pure set $\atoms$, every first-order formula is effectively equivalent to a quantifier-free formula, via a simple quantifier-elimination procedure.
Thus, as long as complexity issues are ignored and decidability is the only concern, in the case of the pure set, we can safely restrict to quantifier-free formulas,
however, in general, definable sets may include arbitrary first-order formulas.

A definable function $f:X\to Y$ is simply 
a function whose domain $X$, codomain~$Y$, and graph 
$\Gamma(f)\subset X\times Y$ are definable sets.
A relational structure is definable if its signature, universe,
and interpretation function
that maps each relation symbol to a relation on the 
universe, are definable\footnote{A structure over a finite signature is definable in $\atoms$ if and only if it is interpretable in $\atoms$, in the model-theoretic sense.}.
Finally, a definable {isomorphism} between definable structures over the same signature is a definable bijective
mapping between their universes that preserves and reflects every relation in the signature. Likewise one introduces, e.g., definable {homomorphisms}.
All hereditarily finite sets (finite sets, whose 
elements are finite, and so on, recursively) are definable, and every finite relational structure over a finite signature is (isomorphic to) a definable one.

\mypar{Contribution} 
The classical {\em isomorphism problem} asks whether two given \emph{finite} structures are isomorphic.
In this paper, we consider its counterpart in the setting of definable sets
(the problem is called \emph{definable isomorphism problem} in the sequel): 
given two definable structures $\str A,\str B$ over the same definable signature $\Sigma$, 
all over the same fixed structure $\Atoms$, are they related by a definable isomorphism?
Note that definable structures can be meaningfully considered as input to a computational problem since they are finitely described using the set-builder notation and first-order formulas in the language of $\atoms$.
The structure $\atoms$ is considered here in a parametric manner, not as a part of input: every structure $\atoms$ induces a different decision problem.

As our main result we prove, under a certain assumptions on the structure $\atoms$,
that the definable isomorphism problem is decidable.
The key technical difficulty is to show that every two 
$S$-definable structures related by a definable isomorphism are also related by an $S$-definable one.
(When $S = \emptyset$ this is parameter elimination: 
existence of an isomorphism defined with parameters enforces existence of one defined without parameters.)
Having this, the problem reduces to testing whether two $S$-definable structures are related by an $S$-definable isomorphism,
which in turn reduces to the first-order satisfiability problem in $\atoms$. 

As witnessed by Example~\ref{ex:nondefiso} below, existence of an isomorphism does not guarantee existence of a definable one.
Therefore we do not solve the \emph{isomorphism problem} for definable structures, which asks whether two given definable structures (over the same signature) are isomorphic.
In fact, the decidability status of the latter problem remains an intriguing open question, even for $\atoms$ being the pure set.

\mypar{Motivation and related work} 
This paper is part of a programme aimed at generalizing classical decision problems 
(for instance the homomorphism problem, studied recently in~\cite{lics2015,KLOT16}),
and computation models such as automata~\cite{BKL14}, Turing machines~\cite{bklt,KLOT14} and programming languages~\cite{BBKL12,BT13,KS16,KT16},
to sets with atoms. For other applications of sets with atoms (called there {\em nominal sets}) in computing, see~\cite{pitts}.

Isomorphism testing is at the core of many decision problems in combinatorics and logic. 
In case of finite graphs it is well known to be solvable in NP, and since recently  in quasi-polynomial time~\cite{Babai16}.
Whether it can be solved in P is still an extremely challenging open question, and only special cases are shown so by now, e.g.~\cite{Luks82}.

\mypar{Acknowledgement}
We are grateful to Pierre Simon for valuable discussions.


%% file: prelim.tex

\section{Preliminaries} \label{sec:atoms}
\colsec{Section~\ref{sec:atoms}}
Throughout the paper, fix a countable relational structure $\atoms$, whose elements are called \emph{atoms}. 
We overload the notation and use the symbol $\atoms$ both for the relational structure and for the set of its elements,
hoping that this does not lead to confusion.
We assume that the vocabulary of $\Atoms$ is finite. 
We shall now formally introduce the notion of definable sets over $\atoms$, following~\cite{KT-POPL, lics2015,KLOT16}.

\mypar{Definable sets}\label{sec:atoms-atoms}
An \emph{expression} is either a variable (from some fixed infinite set or variables), 
or a formal finite union (including the empty union $\emptyset$) of {\em set-builder expressions} of the form 
\begin{align} \label{eq:set-builder}
\setof{e}{a_1,\ldots,a_n\in\Atoms, \phi},
\end{align}
where $e$ is an expression, $a_1,\ldots,a_n$ are pairwise different variables, and
$\phi$ is a first-order formula over the signature of $\atoms$.
The variables $a_1, \ldots, a_n$ are considered bound in $e$ and $\phi$.
The free variables in~\eqref{eq:set-builder} are those free variables of $e$ and of $\phi$ which are not among $a_1,\ldots,a_n$.

For an expression $e$ with free variables $V$, any valuation
$\val:V\to \atoms$ defines in an obvious way 
 a value $X=e[\val]$,
which is either an atom or a set,
formally defined 
by induction on the structure of $e$.
We then say that $X$ is a \emph{definable set over $\atoms$}, and that it is \emph{defined} by $e$ with $\val$.
When the structure $\atoms$ is obvious from the context, 
we simply speak of \emph{definable sets} without explicitly specifying $\atoms$.
Note that one set $X$ can be defined by many different expressions. 
Finally, observe that the family of definable sets is hereditary: every element of a definable set is a definable set (or an atom).

Sometimes we want to emphasize those atoms that appear in the image of the valuation
$\val:V\to \atoms$. 
For any finite set $S \subseteq \atoms$ of atoms with $\val(V)\subseteq S$ 
we say 
that $X$ is \emph{$S$-definable}.
Clearly, an $S$-definable set is also $T$-definable whenever $S\subseteq T$.

As syntactic sugar, we allow
atoms to occur directly in set expressions (these atoms are  called \emph{parameters}). For example, what we write as
the $\set{\atom 1}$-definable set $\set{a\st a\in\atoms, a\neq \atom 1}$ is formally defined 
by the expression $\set{a\st a\in\atoms, a\neq b}$, together with a valuation mapping $b$ to $\atom 1$. 
With this syntactic sugar, a definable set is determined by a sole expression~$e$ with parameters, but without valuation.

As a notational convention, when writing set-builder expressions~\eqref{eq:set-builder} we omit the formula $\phi$ when it is trivially true, and omit the
enumeration $a_1, \ldots, a_n \in \atoms$ when $n=0$.
This allows us, in particular, to write singletons, like $\set{\atom 1}$.

\begin{remark}\label{rem:conventions}\rm
To improve readability, it will be convenient to use standard set-theoretic
encodings to allow a more flexible syntax.
In particular, ordered pairs and tuples can be encoded e.g.~by Kuratowski pairs, $(x,y)=\{\{x,y\},\{x\}\}$.
We will also consider definable infinite families 
of symbols, such as $\set{R_x: x \in X }$, where $R$ is a symbol and $X$ is a definable set. Formally, such a family can be encoded as the set of ordered pairs $\set{R}\times X$, where 
the symbol $R$ is represented by some \0-definable set, e.g. $\emptyset$ or $\set{\emptyset}$. 
Here we use the fact that definable sets (over any fixed atoms) are closed under Cartesian products.
\end{remark}

\mypar{Definable relational structures}
Any object in the set-theoretic universe (a relation, a function, a relational structure, etc.) may be definable. For example, a definable relation on $X,Y$ is a relation $R\subset X\times Y$
which is a definable set of pairs, and a definable function $X \to Y$ is a  function whose graph is definable.
Along the same lines, a definable relational signature is a definable set of \emph{symbols} $\Sigma$,
partitioned into definable sets $\Sigma=\Sigma_1\uplus\Sigma_2\uplus\ldots\uplus\Sigma_l$ according to the arity of symbols.
We say that $\sigma$ has \emph{arity} $r$  if $\sigma\in \Sigma_r$, and
$l\in\mathbb{N}$ is thus the maximal arity of a symbol in $\Sigma$.

For a signature $\Sigma$, a definable $\Sigma$-structure $\str A$ consists of a definable universe $A$
and a definable interpretation function which
assigns a relation $\sigma^{\str A}\subset A^r$ to each relation symbol $\sigma\in\Sigma$ of arity~$r$.
(We denote structures using blackboard font, and their  universes using the corresponding symbol in italics).
More explicitly, such a structure can be represented by the tuple $\str A = (A, I_1, \ldots, I_l)$
where $I_r=\set{(\sigma,a_1,\ldots,a_r)\st \sigma \in\Sigma_r, (a_1,\ldots,a_r)\in \sigma^{\str A}}$
is a definable set for $r=1,\ldots, l$. 
It is not difficult to see that the interpretation $\sigma^{\str A}$ of every fixed symbol $\sigma \in \Sigma$ is definable.



\begin{remark}\label{rem:inter}\rm
    As argued in~\cite{KLOT16}, definable structures over finite signatures coincide, up to definable isomorphism, 
    with \emph{first-order interpretations with parameters} in $\atoms$, in the sense of model theory~\cite{hodges}.
    \qed
    \end{remark}

\begin{example}\label{ex:expressions}
The graph $G$ from Example~\ref{ex:intro} is a definable (over $\atoms$ being the pure set)
structure over a finite signature $\Sigma$ containing a single binary relation symbol.
To give an example of a definable structure over an infinite definable signature, extend $G$ to a structure
$\str A$ by infinitely many unary predicates representing 
the neighborhoods of each vertex of $G$.
To this end, define the signature 
$\Sigma=\set{E}\cup \setof{N_v} {v\in V}$, where $V=\setof{\set{a,b}}{a,b\in\atoms, a\neq b}$ is the vertex set of $G$ and $N$ is a symbol (cf. Remark~\ref{rem:conventions}).
The interpretation of $N_v$
is specified by the set $I_1=\setof{(N_v,w)}{(v,w)\in E}$
(where $E$ is defined by the expression from Example~\ref{ex:intro}).  
\qed
%
%
%
\end{example}

\mypar{Representing the input}
Definable relational structures can be input to algorithms, as they are finitely presented by expressions defining 
the signature, the universe, and the interpretation function. If the input is an $S$-definable set~$X$ defined by an expression $e$
with parameters $a_1, \ldots, a_n \in S$, then we also need to represent the tuple $a_1,\ldots, a_n$ of atoms. 
For example, in the case of pure set $\atoms$ these elements can be represented as arbitrary pairwise distinct numbers. 

\mypar{Definable isomorphism problem}
Let the structure of atoms be fixed and denoted by $\atoms$.
Recall that a definable function $f:X\to Y$ is a function whose graph $\Gamma(f)\subset X\times Y$ is a definable set.
A definable {isomorphism} between definable structures $\str A$, $\str B$ over the same signature $\Sigma$ is a definable bijective
function $h : A \to B$ between their universes that preserves and reflects every relation in the signature: for every $\sigma \in \Sigma$ of arity $r$
and every $r$-tuple $a_1, \ldots, a_r \in A$ of elements of $\str A$,
$(a_1, \ldots, a_r) \in \sigma^{\str A}$ if, and only if $(h(a_1), \ldots, h(a_r)) \in \sigma^{\str B}$.
Likewise one can also introduce definable {homomorphisms}, embeddings, etc.

We focus in this paper on the following family of decision problems 
(note that the structure $\atoms$ is fixed, and not part of input, and hence every choice of $\atoms$ yields a 
different decision problem):

\decproblemdef{Definable-isomorphism($\atoms$)}
{A definable signature $\Sigma$ and two definable $\Sigma$-structures $\str A$ and $\str B$.}
{Is there a definable isomorphism from $\str A$ to $\str B$?}
\vspace{2mm}

\begin{example} \label{ex:nondefiso}
Imposing the definability requirement on isomorphisms clearly does matter.
Let $\atoms$ be the pure set again, and
consider the following two $\emptyset$-definable graphs, each of them being an infinite, edgeless graph: 
\begin{align*}
V_1 &=  \setof{a}{a\in\atoms} = \atoms & V_2&=\setof{\set{a,b}}{a,b\in\atoms, a\neq b},\\
E_1 &= \emptyset & 
E_2 &=\emptyset.
\end{align*}
The two graphs are clearly isomorphic. On the other hand there is no \emph{definable} isomorphism between them,
simply because there is no definable bijection between $V_1$ and $V_2$,
as we will argue in Example~\ref{ex:cont} in Section~\ref{sec:orbitfinite}.
%
\qed
\end{example}


%% file: defiso.tex

In this paper, we will always assume that the structure $\atoms$ is $\omega$-categorical.
We will also make assumptions regarding computability properties, least supports and denseness, as introduced below.


\mypar{$\omega$-categoricity}
We say that a countable structure $\atoms$ is \emph{$\omega$-categorical}
if any countable structure $\atoms'$ which satisfies the same first-order sentences as $\atoms$ is isomorphic to $\atoms$.
The following fundmanental theorem, due to  Ryll-Nardzewski, Engeler and Svenonius~\cite{Eng59,RN59,Sven59}, gives a useful characterization of $\omega$-categorical structures in terms of their automorphism groups. Below, automorphisms act
on the set $\atoms^n$ of all $n$-tuples of elements of $\atoms$ in a coordinatewise fashion.
For each $n\ge 1$, the set $\atoms^n$ is partitioned into the orbits of this action.

\begin{theorem}\label{thm:RN}
 A structure $\atoms$ is $\omega$-categorical if, and only if 
for each $n$,  the automorphism group of $\atoms$ induces 
finitely many orbits on the set $\atoms^n$.  Moreover, if $\atoms$ is $\omega$-categorical, then each orbit of this action can be defined by a first-order formula  with $n$ free variables.
\end{theorem}

Examples of $\omega$-categorical structures include:
\begin{itemize}
	\item the pure set;
	\item the dense total order $(\mathbb{Q}, \leq)$ of rational numbers;
	\item the universal (random) graph (the Fra\"iss\'e limit of all finite graphs~\cite{Fraissebook});
	\item the universal partial order (the Fra\"iss\'e limit of all finite partial orders);
	\item a countable vector space over a finite field.
	\end{itemize}
	All the structures considered in this paper as atoms are assumed to be $\omega$-categorical.

	\mypar{Effectivity}
We will additionally impose certain computability assumptions on $\atoms$.
First, we fix an encoding of elements of $\atoms$ as strings,
i.e., a surjection from $\set{0,1}^*$ to $\atoms$.
Similarly, we fix an encoding of the symbols in the signature of $\atoms$ as strings.
This allows us to represent first-order formulas, together with valuations of their free variables, as strings.
We then assume that $\atoms$ has a decidable first-order theory,
i.e., there is an algorithm which inputs a first-order formula, together with a valuation of its free variables in $\atoms$, and determines whether the valuation satisfies the formula in $\atoms$. 
We remark that decidability of the first-order theory of $\atoms$ implies (and is equivalent to) the existence of an algorithm which inputs two expressions describing definable sets $x$ and $y$, and decides whether $x=y$.
Lastly, we assume \emph{computability of the Ryll-Nardzewski function},
which maps a given number $n$ to the number of orbits of $\atoms^n$.
We say that a structure $\atoms$ is \emph{effectively $\omega$-categorical}
if it is $\omega$-categorical, has a decidable first-order theory, and its Ryll-Nardzewski function is computable.
Note that effective $\omega$-categoricity implies that there is an algorithm which inputs a number $n$ and 
outputs the orbits of $\atoms^n$, each defined by a first-order formula with $n$ free variables.

\mypar{Automorphisms, partial automorphisms, and self-embeddings}
Automorphisms of $\atoms$ will be called \emph{atom automorphisms}.
For a finite $S\subseteq\atoms$, \emph{atom $S$-automorphisms} are those atom automorphisms $\pi$ which fix all elements of $S$, i.e., $\pi(a) = a$ for every $a\in S$.
If $\pi$ is an atom automorphism and $x$ is a definable set,
defined by an expression $e$ with parameters $a_1,\ldots,a_n$,
then we define the set $\pi x$ as the set defined by the same expression,
and the parameters $\pi(a_1),\ldots,\pi(a_n)$. 

A \emph{partial automorphism}\footnote{also called a 
\emph{partial elementary map}.} of $\atoms$ is a partial bijection $f$ 
between two subsets of $\atoms$, such that 
for every first-order formula $\phi(x_1,\ldots,x_k)$
and every tuple $a_1,\ldots,a_k\in \dom f$,
$\phi(a_1,\ldots,a_k)$ holds in $\atoms$ if, and only if
$\phi(f(a_1),\ldots,f(a_k))$ holds in $\atoms$.
In particular, an automorphism is a partial automorphism whose domain and codomain are $\atoms$.
A \emph{self-embedding} of $\atoms$
is a partial automorphism whose domain is $\atoms$.

A partial automorphism is \emph{finite} if its domain is finite.
The following lemma is a consequence of the theorem of Ryll-Nardzewski, Engeler, and Svenonius, specifically of the fact that two $n$-tuples of atoms are in the same orbit if and only if they satisfy the same first-order formulas. 

\begin{lemma}\label{lem:extend}
	Every finite partial automorphism of an $\omega$-categorical structure extends to an automorphism of that structure.
\end{lemma}

\mypar{Least supports}
We say that a finite set $S\subset\atoms$ \emph{supports} $x$ if every $S$-automorphism $\pi$ fixes $x$, i.e., $\pi x=x$. Note that  any $S$-definable set $x$ is supported by $S$.
We say that $\atoms$ has \emph{least supports}
if every definable set $x$ has a least (under inclusion) support.
On the face of it, admitting least supports is a rather complex condition\footnote{It can be shown to be equivalent to the conjunction of \emph{weak elimination of imaginaries} and \emph{trivial algebraic closure}, which are well-studied notions in model theory.}, in that its formulation relies on the notion of definable sets. However, for $\omega$-categorical structures the existence of least supports has an equivalent structural characterization in terms of $\atoms$ itself, see in~\cite[Thm.~9.3]{BKL14}.
Examples of atoms with least supports include:
\begin{itemize}
	\item  the pure set,
	\item the dense total order $(\mathbb Q,\le)$,
	\item the random graph,
	\item the universal homogeneous partial order.
\end{itemize}
\newcommand{\K}{\mathbb K}
\begin{example}\label{ex:vect}
	An example which does not have least supports is 
	an infinite-dimensional vector space over a finite field.
The idea is that any basis of a finite-dimensional subspace is a support of that subspace,
	but there is no least support.

More precisely, fix a finite field $\K$, e.g. the two-element field.
Let $V$ be a countable vector space over $\K$.
A concrete example  is  
obtained by considering the set of all infinite sequences of elements of $\K$,
with finitely many non-zero elements, with coordinatewise addition and multiplication by scalars.

We treat $V$ as a logical structure $\str V$ equipped with 
the binary addition operation $x,y\mapsto x+y$ and 
unary operations $x\mapsto c\cdot x$, for each scalar $c\in\K$. 
Then, automorphisms of $\str V$ correspond precisely to 
 invertible linear maps of $V$. 
Clearly, the vector space $V$ has countable dimension.
By basic results of linear algebra, any two countable-dimensional
vector spaces over $\K$ are isomorphic. Hence, $\str V$ is an $\omega$-categorical structure.


We argue that the structure $\str  V$, treated as atoms, does not have least supports. 
Fix $k\ge 0$ and fix a $k$-dimensional subspace $W$ of $V$.
Then $W$ is a definable set, as it is finite
(of cardinality $|\K|^k$).


We claim that a finite set $S\subset \atoms$ is a support of $W$
if, and only if, $W$ is contained in the linear span of $S$. 
In one direction, suppose that $W$ is contained in the linear span of $S$.
In particular, every element $w$ of $W$ can be expressed as a linear combination of 
elements of $S$. Then, every linear automorphism of $V$ which fixes $S$ pointwise must fix $w$.
Since this applies to all $w\in W$, $W$ is fixed by all $S$-automorphisms of $\str V$.

Conversely, suppose that $W$ is not contained in the linear span of $S$,
i.e., there is some element $w\in W$ which is not a linear combination of the elements of $S$.
Then there is a linear automorphism of $V$ which is the identity on $S$ and maps $w$ to 
some element $w'$ not in $W$ (we use the fact that $V$ has infinite dimension and $S\cup W$ is finite).

In particular, a set $S\subset \str V$ is a  minimal support of $W$ if and only if
$S$ is a basis of $W$. 
Now, if $k>1$ then every space $W$ of dimension $k$ has more than one basis\footnote{If $e_1,e_2,\ldots,e_k$ is a basis then so is $e_1+e_2,e_2,\ldots,e_k$.
If $k=1$ then $e_1$ and $-e_1$ are distinct bases, unless $\K=2$.
}.
In particular, $W$ has multiple minimal supports, and hence has no least support. 
\end{example}

\mypar{Denseness}
We call the structure $\atoms$  \emph{dense} if 
for all finite sets $T$ and $S$, where $T\subset S\subseteq \atoms$, there is a self-embedding $H$ which fixes $T$ pointwise, such that every automorphism of $H(\atoms)$ extends to an atom $S$-automorphism.
All the structures mentioned above are dense.



\begin{samepage} 
\section{Parameter elimination}
The main property that will ensure the decidability of the definable isomorphism problem  is the following \emph{parameter elimination property}:
\begin{quote}\itshape
If two $T$-definable relational structures $\str A$ and $\str B$ are related by a definable isomorphism, then they are also related by a $T$-definable one.
\end{quote}
\end{samepage}
Roughly speaking, the lemma says that in a definition of an isomorphism between $T$-definable structures, 
one can eliminate parameters outside of $T$, possibly at the price of modifying the isomorphism. The following example shows that even in very simple situations, constructing a $T$-definable isomorphism from a given definable one may be nontrivial.

\begin{example}\label{ex:bijection-smooting}
	Let $\atoms$ be the pure set, put $T=\emptyset$, and consider sets ({\em qua} relational structures over the empty vocabulary):
	\[
		\str A = \str B = \atoms^2 + \atoms.
	\]
	Fix an atom $c\in\atoms$, and define $f:\str A\to \str B$ by:
	\[
		f(a) = (a,c), \qquad f(a,c) = a, \qquad f(a,b)=(a,b)
	\]
	for every $a\in\atoms$ and $b\in\atoms\setminus\{c\}$. This is a $\{c\}$-definable bijection between $\str A$ and $\str B$ which needs to be ``smoothed out'' to yield a $\emptyset$-definable one (in this case, e.g. the identity function). Note that this requires altering $f$ even at some arguments that do not contain $c$ in their supports.
		\end{example}

		To prove  parameter elimination, even for atoms being the pure set, in Section~\ref{sec:elim} we will provide an iterative procedure for ``smoothing out isomorphisms'' by gradual elimination of parameters. This is the main technical result of this paper.
More precisely, we prove: 
\begin{theorem}\label{thm:elim}
	Any dense, $\omega$-categorical structure $\atoms$  with least supports 
	has the parameter elimination property.
\end{theorem}
Theorem~\ref{thm:elim} is proved in Section~\ref{sec:elim}.
We now show how the parameter elimination property yields decidability: 
\begin{lemma}\label{lem:to-par-el}
	Suppose that $\atoms$ is effectively $\omega$-categorical 
	and has the parameter elimination property. Then \problem{Definable-isomorphism($\atoms$)} is decidable.
\end{lemma}

\begin{proof}
	By the parameter elimination property, \problem{Definable-isomorphism($\atoms$)} reduces to testing whether
given $T$-definable $\Sigma$-structures $\str A$, $\str B$ are related by a $T$-definable isomorphism.
In turn, as we show now, testing of the latter condition reduces to evaluation of first-order formulas in $\atoms$.
Let the given structures be $\str A = (A, I_1, \ldots, I_l)$ and $\str B = (B, J_1, \ldots, J_l)$.
We follow the lines of the proof of Thm.~12 in~\cite{lics2015}; in particular,
we build on the following fact, which follows from effective $\omega$-categoricity of $\atoms$ (cf.~\cite[Lemma~5.27]{atombook}): 
\begin{lemma}\label{lem13}Suppose that $\atoms$ are effectively $\omega$-categorical.
For any finite set $T$ of atoms,
a $T$-definable set $X$ has only finitely many $T$-definable subsets, and expressions defining them can be computed effectively
from an expression defining $X$.
\end{lemma}
%

To verify existence of a $T$-definable isomorphism from $\str A$ to $\str B$, apply Lemma~\ref{lem13} to
$X = A \times B$ and for every $T$-definable subset $R\subseteq A\times B$, test the validity of the first-order
formula
\[
\forall a\in A \ \exists ! b\in B \  R(a,b) \ \wedge \ 
\forall b\in B \ \exists ! a\in A \  R(a,b)
\]
ensuring that $R$ is the graph of a bijection; and for every $i = 1,\ldots, l$, test the validity of the fomula
\begin{align*}
\begin{array}{r}
\forall \sigma \in \Sigma_i \ 
\forall a_1, \ldots, a_i \in A \\  \forall b_1, \ldots, b_i \in B 
\end{array}
\quad
\bigwedge_{1\leq j \leq i} R(a_i, b_i) \implies 
\big(I_i(\sigma, a_1, \ldots, a_i) \iff J_i(\sigma, b_1, \ldots, b_i)\big)
\end{align*}
ensuring that the function is an isomorphism.
Evaluation of first-order formulas of the above form reduces to evaluation of first order formulas in $\atoms$, see~\cite{BT13,lics2015} for further details.
\end{proof}

\medskip


Theorem~\ref{thm:elim} and Lemma~\ref{lem:to-par-el} together prove:
\begin{theorem} \label{thm:decid}
\problem{Definable-isomorphism($\atoms$)} is decidable
whenever $\atoms$ is 
an effective $\omega$-categorical structure which is dense and has least supports.
\end{theorem}
%

\begin{remark}\rm
In this paper we consider purely relational signatures, but nothing changes if function symbols are allowed:
the proofs of Theorem~\ref{thm:elim} and Lemma~\ref{lem:to-par-el} (and thus of Theorem~\ref{thm:decid}) 
are still valid when some signature symbols are enforced to be interpreted as functions.
\end{remark}

\section{On the necessity of the assumptions}

Before proving Theorem~\ref{thm:elim},
we review the assumptions it uses: denseness and least supports.
Note that effective $\omega$-categoricity is not used in Theorem~\ref{thm:elim}, but only in the proof of Theorem~\ref{thm:decid}, 
i.e., in Lemma~\ref{lem:to-par-el} (we do, however, always assume $\omega$-categoricity).
\medskip



An essential assumption that makes Theorem~\ref{thm:elim} go through is the denseness of $\atoms$. Dropping it would invalidate the lemma, in view of the following counterexample:

\begin{example}\label{ex:circle}
Let $\atoms$ be the set of rational numbers with the so-called {\em cyclic order} relation, which is a ternary relation $R$ defined by:
\[
	R(a,b,c) \quad\text{if and only if}\quad a<b<c \quad \text{or} \quad b<c<a \quad\text{or}\quad c<a<b.
\]
(Note that the binary order relation itself is not in the vocabulary of $\atoms$.) These atoms can be visualized as densely distributed on an oriented circle, so that one cannot say whether one atom is ``smaller'' or ``greater'' than another, but one can say whether three atoms $a,b,c$ follow each other in the clockwise direction:
\begin{equation}\label{eq:circle}
\xy 
{\ellipse<30pt,30pt>{-}};
p+(-10.5,-.5)*\dir{*}="o",*+!UR{a};
p+(4.5,9.0)*\dir{*}="o",*+!DR{b};
p+(15.3,-4.0)*\dir{*}="o",*+!DL{c};
p+(-0.4,-10.0)*\dir{*}="o",*+!UL{d};
\endxy
\qquad
\begin{array}{ll}
\text{e.g.} \\
& (a,b,c)\in R,  (b,d,a) \in R \\
& (a,c,b)\not\in R, (b,d,c)\not\in R
\end{array}
\end{equation}
For example, in the drawing above, there is an atom automorphism that maps the ordered pair $(a,b)$ to $(c,d)$ (that is, it maps $a$ to $b$ and $c$ to $d$), there is also one that maps $(a,b)$ to $(d,c)$ and one that maps $(a,b,c)$ to $(b,c,d)$, but there is no atom automorphism that maps $(a,b,c)$ to $(a,d,c)$.

This structure of atoms is effectively $\omega$-categorical
(in fact, it is definable in $(\mathbb Q,\le)$).
 It is also not difficult to check, using the criterion of~\cite[Thm.~9.3]{BKL14}, that is admits least supports. However, it is not dense. To see this, fix some $a\in\atoms$ and put $T = \emptyset$ and $S=\{a\}$. Let $H$ be an embedding of $\atoms$ into $\atoms$ that avoids $a$, and pick some $b\neq c$ in $H(\atoms)$, i.e. the image of $H$. Since $H(\atoms)$ is isomorphic to $\atoms$, there is some automorphism of it that swaps $b$ and $c$. This automorphism does not extend to any $\{a\}$-automorphism of $\atoms$, since no $\{a\}$-automorphism can swap $b$ and $c$. 

We shall now show that parameter elimination fails for the structure $\atoms$. Put $T=\emptyset$, and define a directed graph $\str A$  so that:
\begin{itemize}
\item the set of vertices is the relation $R$; more explicitly, vertices are ordered triples $(a,b,c)\in\atoms^3$ such that $(a,b,c)\in R$, 
\item from each vertex $(a,b,c)$ there is exactly one directed edge, ending in the vertex $(b,c,a)$.
\end{itemize}
As a directed graph, $\str A$ is an infinite disjoint family of directed triangles.

Furthermore, define a graph $\str B$ as obtained from $\str A$ by inverting all arrows; that is, from each vertex $(a,b,c)$ in $\str B$ there is exactly one directed edge, ending in the vertex $(c,a,b)$.

Clearly, $\str A$ and $\str B$ are isomorphic as graphs. However, no $\emptyset$-definable isomorphism between them exists. Indeed, such an automorphism $\pi$ would have to map a vertex $(a,b,c)$ to one of three candidates: $(a,b,c)$, $(b,c,a)$ or $(c,a,b)$. If the first is chosen then, by $\emptyset$-definability, $\pi$ has to map $(b,c,a)$ to $(b,c,a)$, but on the other hand, since $\pi$ must preserve edges, it must map $(b,c,a)$ to $(c,a,b)$, which is a contradiction. The other two candidates for $\pi(a,b,c)$ are excluded by analogous arguments.

On the other hand, for any fixed atom $d\in\atoms$, a $\{d\}$-definable isomorphism $\pi:\str A\to\str B$ exists and is defined as follows. For any set $\{a,b,c\}\subseteq \atoms$, we shall define a $\{d\}$-definable isomorphism between directed graphs
\[
\vcenter{\xymatrix{
\save[]+<30pt,0pt>*{(a,b,c)}\ar[dr]\ar@{<-}[d]\restore \\
(c,a,b) & (b,c,a)\ar[l]
}}
\qquad\text{and}\qquad
\vcenter{\xymatrix{
\save[]+<30pt,0pt>*{(a,b,c)}\ar[dr]\ar@{<-}[d]\restore \\
(b,c,a) & (c,a,b),\ar[l]
}}
\]
which are fragments of $\str A$ and $\str B$, respectively.

First, consider the case where $d\not\in\{a,b,c\}$. Then exactly one of the triples
\[
	(a,d,b), (b,d,c), (c,d,a)
\]
belongs to the relation $R$. (This becomes clear by looking at the drawing~\eqref{eq:circle} above, where $(c,d,a)\in R$.) One of the triples $(a,b,c)$, $(b,c,a)$ and $(c,a,b)$ is thus singled out by $d$, and the isomorphism of the above triangles that fixes that triple and swaps the other two, is $\{d\}$-definable.

If on the other hand $d\in\{a,b,c\}$, then one of the triples $(a,b,c)$, $(b,c,a)$ and $(c,a,b)$ is even more directly singled out by $d$, and a $\{d\}$-definable automorphism is defined as before.

Altogether we have constructed $\emptyset$-definable graphs $\str A$ and $\str B$ that are related by a definable isomorphism, but not by a $\emptyset$-definable one.
\qed
\end{example}

We now discuss the necessity of the assumption on least supports.
Removing this assumption also causes Theorem~\ref{thm:elim} to fail,
as witnessed below.
\begin{example}
	Fix a finite field $\K$.
	Let $\str A$ be a countable affine space over $\K$.
	This can be defined for example as follows.

	Let $V$ be a countable vector space over $\K$, cf. Example~\ref{ex:vect}.
	Intuitively, $\str A$ is the vector space $V$ with the origin point forgotten. Formally, for each scalar $c\in\K$, let
	$R_c\subset V^4$ consist of all tuples $(a_1,a_2,b_1,b_2)\in V^4$ with $a_1-a_2=c\cdot (b_1-b_2)$.
	We call the structure $\str A$ with domain $V$ and relations $R_c$, for $c\in\K$, the \emph{affine space} modelled on $V$.
	The fundamental property of $\str A$ is that 
	its automorphisms are the \emph{affine automorphisms}, i.e.,  the permutations  of $V$ which are of the form $a\mapsto b_0+ \pi(a-a_0)$, for some $a_0,b_0\in \str A$ and some linear automorphism $\pi$ of $V$.

	Since $\str A$ is definable in terms of the $\omega$-categorical structure 
	$\str V$,  by Theorem~\ref{thm:RN}, $\str A$ itself is $\omega$-categorical\footnote{Indeed, all automorphisms of $\str V$ remain automorphisms of $\str A$, so $\str A^n$ has at most as many orbits under the action of affinie automorphisms as $\str V^n$ has under the action of linear automorphisms, i.e., finitely many.}.
It is also not difficult to check that the structure $\str A$ is dense.

By a similar argument as in Example~\ref{ex:vect}, the structure $\str A$ does not have least supports. 
Indeed, any affine line contained in the affine space is supported by any two of its points.
(For the special case $\K=\mathbb{Z}_2$ this is not a counterexample yet, as an affine line in this case consists of only two points. In this case, a counterexample is obtained by looking at planes: a two-dimensional plane is supported by any three of its points.)

We now show that the structure $\str A$, treated as atoms, does not have parameter elimination. To this end, we show two $\emptyset$-definable sets over $\str A$ such that there is a definable bijection between them, but no $\emptyset$-definable bijection. 

The first set is the domain $A$ of $\str A$ itself.
The second set $B$ is essentially (an isomorphic copy of) the vector space $V$ underlying $\str A$.
Formally, it is defined as the set of equivalence classes of pairs $(a,b)\in A^2$, where $(a,b)$ is equivalent to $(a',b')$ iff $b-a=b'-a'$. This equivalence relation is clearly definable. The set $B$ has a $\emptyset$-definable element $0$, namely the equivalence class of a pair $(a,a)$.
The set $A$ has no $\emptyset$-definable point, since 
such a point would be invariant under all automorphisms of $\str A$,
and affine automorphisms act transitively on $A$.  
As any $\emptyset$-definable bijection  $f\from B\to A$ would map 
the $\emptyset$-definable point $0\in B$ to a $\emptyset$-definable point $f(0)\in A$,
no such bijection can exist.

On the other hand, for any fixed $a\in A$, 
there is a $\set{a}$-definable bijection between $A$ and $B$, 
namely the function mapping $b\in A$ to the equivalence class of the pair $(a,b)$.\qed

\end{example}

%% file: orbitfinite.tex
\section{Definable sets via the action of atom automorphisms} \label{sec:orbitfinite}

For the proof of Theorem~\ref{thm:elim} it will be more convenient to take a different perspective on definable sets,
namely via the action of atom automorphisms. This view emphasises that definable sets are always \emph{orbit-finite}
(cf.~Lemma~\ref{lem:orbitfinite} below). 
In this section we provide the necessary definitions and properties that will be useful
in the proof of Theorem~\ref{thm:elim}  in the next section.
All further missing details can be found in~\cite{atombook}.

Definable sets contain, as elements, either other definable sets or atoms $a\in \atoms$.
The group of atom automorphisms acts naturally on such sets, by renaming all atoms appearing as elements, as elements of elements, etc.
The action preserves definable sets: a definable set is mapped to a definable set.
By $\pi x$ we denote the result of the action of an atom automorphism $\pi$ on a definable set $x$.
For instance, consider the pure set $\atoms$ as atoms and the atom automorphism $\pi$ 
that swaps $\atom 0$ with $\atom 1$, and $\atom 3$ with $\atom 4$, and preserves
all other atoms. Then 
\[ \pi   \setof{a}{a\in\atoms,a\neq \atom 1\land a\neq \atom 2} 
\  = \ 
 \setof{a}{a\in\atoms,a\neq \atom 0\land a\neq \atom 2}
 \quad
 \pi  \{\atom 0, \atom 1, \atom 2\} \  = \  \{\atom 0, \atom 1, \atom 2\}. 
 \]
The action defines a partition of all definable sets into \emph{orbits}: $x$ and $x'$ are in the same orbit if $\pi x = x'$ for some atom automorphism $\pi$.
In the same vein, for every finite set $S\subseteq \atoms$, 
the action of the subgroup of atom $S$-automorphisms defines a finer partition of all definable sets into \emph{$S$-orbits} 
($\emptyset$-orbits are just orbits).
In particular, the set $\atoms$ of atoms is itself also partitioned in $S$-orbits.
%

%

By inspecting the syntactic form of definable sets, one easily verifies the following basic fact:
\begin{lemma} \label{lem:closed}
Every $S$-definable set is closed under the action of atom $S$-automorphisms on its elements, i.e., 
is a union of $S$-orbits.
\end{lemma}
We will intensively use the following consequence of Lemma~\ref{lem:closed}:
\begin{lemma} \label{lem:equivfunc1}
Every $S$-definable function $h$ commutes with atom $S$-automorphisms: for every atom $S$-automorphism $\pi$, 
$
h\pi = \pi h.
$
\end{lemma}
%
It is not difficult to prove, by induction on the structure of set-builder expressions, that $\omega$-categoricity of $\atoms$ guarantees finiteness in the statement of Lemma~\ref{lem:closed}:
\begin{lemma} \label{lem:orbitfinite}
Assume that the structure of atoms is $\omega$-categorical. Then
every $S$-definable set is a finite union of $S$-orbits.
\end{lemma}
%
%
For an $S$-orbit $O$ and an $S$-definable set $X$, if $O \subseteq X$ we say that $O$ is an $S$-orbit \emph{inside} $X$; on the other hand,
whenever $x\in O$ we say that $O$ is \emph{the $S$-orbit of} $x$, and write $O = \orbit{S}{x}$.
The converse of Lemma~\ref{lem:equivfunc1} is also true when $\atoms$ is $\omega$-categorical,
as follows from Theorem~\ref{thm:RN}:
\begin{lemma} \label{lem:equivfunc2}
Assume that the structure of atoms is $\omega$-categorical. Then
every function $h$ that commutes with atom $S$-automorphisms, with $\dom(h)$ and $\codom(h)$ being $S$-definable, is itself $S$-definable.
\end{lemma}

\begin{example} \label{ex:cont}
Lemma~\ref{lem:equivfunc1} can be used to prove the claim formulated in Example~\ref{ex:nondefiso} in Section~\ref{sec:atoms}:
there is no definable bijection between the following two sets
\begin{align*}
V_1 &=  \setof{a}{a\in\atoms} = \atoms & V_2&=\setof{\set{a,b}}{a,b\in\atoms, a\neq b}.
\end{align*}
Suppose the contrary, and let $f : V_1 \to V_2$ be an $S$-definable bijection.  
Consider any $a, b, c\in \atoms - S$ with $f(a) = \set{b, c}$, and any atom $S$-automorphism $\pi$ that swaps 
$b$ and $c$ and does not preserve $a$, say $\pi(a) = a'$.
Note that such $\pi$ always exists, e.g., when $a=b$ then $a'=c$.
By Lemma~\ref{lem:equivfunc1} we obtain:
$
f(a') = f \pi (a) = \pi f(a) = \pi \set{b,c} = \set{b, c} = f(a),
$
which is in contradiction with bijectivity of $f$.
\qed
\end{example}

We will also need some basic properties involving least  supports.
Recall that we assume that the structure $\atoms$ has least supports, and that
$\supp x\subset \atoms$ denotes the least support of $x$.
First, we observe that the support function commutes with atom automorphisms.

\begin{lemma} \label{lem:eqsup}
\label{lem:supp}
For every definable set $x$ and atom automorphism $\pi$, we have $\supp{\pi x} = \pi (\supp{x})$.
\end{lemma}
%
%
The cardinality of $\supp{x}$ will be  called the \emph{dimension} of $x$.
\begin{corollary} \label{cor:dim}
Every two elements in the same $\emptyset$-orbit have supports of  same dimension.
\end{corollary}
%
Moreover, the action of an atom automorphism $\pi$ on $x$ depends only on the restriction of $\pi$ to
$\supp{x}$:
\begin{lemma} \label{lem:eqaction}
If atom automorphisms $\pi, \pi'$ coincide on $\supp{x}$, then $\pi x = \pi' x$.
\end{lemma}
We observe the relationship between the least support of a function, its argument and value:
\begin{lemma} \label{lem:supfun}
Let $f$ be a definable function and let $a\in\dom(f)$. Then $\supp{f(x)} \subseteq \supp f \cup \supp x$.
\end{lemma}

    




%% file: elim.tex

\section{Proof of Theorem~\ref{thm:elim}}  \label{sec:elim}

In this section we prove Theorem~\ref{thm:elim}.
%
%
Consider two $T$-definable structures $\str A$ and $\str B$ related by an $S$-definable isomorphism $f$ and assume,
w.l.o.g., that $S \supseteq T$. 
We are going to modify suitably the isomorphism in order to obtain a possibly different one which will be $T$-definable. 
For the sake of readability we elaborate the proof in a special but crucial case, under the following assumptions:
\begin{itemize}
\item $T = \emptyset$, 
\item the signature contains just one binary symbol.
\end{itemize}
Thus we assume the structures $\str A$, $\str B$ to be $\emptyset$-definable directed graphs.
The proof adapts easily to the general case, as discussed at the end of this section.

Consider therefore two $\emptyset$-definable directed graphs $\str A = (A, E)$ and $\str B = (B, F)$, 
where $A$ and $B$ are sets of nodes, and $E\subseteq A\times A$ and $F \subseteq B\times B$ are sets of directed edges,
together with an $S$-definable isomorphism $f : A\to B$ of graphs. 
Assume w.l.o.g.~that the node sets $A$ and $B$ are disjoint.
By Lemma~\ref{lem:orbitfinite},  
the set $A$ of nodes of $\str A$, being itself $\emptyset$-definable, splits into finitely many $\emptyset$-orbits.
Likewise the set $B$ of nodes of $\str B$ splits into finitely many $\emptyset$-orbits,
but a priori it is not clear whether the numbers of $\emptyset$-orbits inside $A$ and $B$ are equal. 
As a side conclusion of our proof, it will be made clear that they really are.
%

We apply denseness of $\atoms$.
Fix in the sequel an embedding $H : \atoms \to \atoms$ such that every automorphism of $H(\atoms)$ extends
to an  $S$-automorphism of $\atoms$.
%
%
Atoms in $H(\atoms)$ will be called \emph{$S$-independent}.
Along the same lines, a node $x \in A\cup B$ will be called \emph{$S$-independent} if $\supp{x} \subseteq H(\atoms)$.
%
%
By $\omega$-categoricity and denseness of atoms we have:
\begin{claim} \label{claim:orbit}
Every $\emptyset$-orbit inside $A \cup B$ contains an $S$-independent node.
\end{claim}
\noindent
(The claim, as well as few other claims formulated below, will be proved below once the proof of Theorem~\ref{thm:elim} is outlined.)

We are going to construct an $\emptyset$-definable bijection $h : A\to B$ which will be later shown to be an isomorphism; 
to this end we will define inductively a sequence of $\emptyset$-definable partial bijections
\[
h_i : A \to B,
\]
for $i = 0, 1, \ldots, m$, where $m$ is the number of $\emptyset$-orbits inside $A$ (or $B$, as made explicit below),
such that the domain $\dom(h_{i+1})$ of every $h_{i+1}$ extends 
the domain $\dom(h_i)$ of $h_i$ by one $\emptyset$-orbit inside $A$.
The required bijection will be $h = h_m$.
The order of adding $\emptyset$-orbits to the domain of $h$ will be relevant for showing that $h$ is an isomorphism.

We start easily by taking as $h_0$ the empty function.
For the induction step, suppose that $h_n$ is already defined. 
Among the remaining $\emptyset$-orbits of $A$ and $B$, i.e., among those which are not included in 
$\dom(h_n) \cup \codom(h_n)$, choose an orbit
$O$ whose elements have maximal dimension (cf.~Corollary~\ref{cor:dim}).
W.l.o.g.~assume that $O \subseteq A$ (if $O\subseteq B$, change the roles of $\str A$ and $\str B$, and replace 
$f$ and all $h_i$ for $i\leq n$ by their inverses).
Choose an arbitrary $S$-independent node $x_0 \in O$.
We observe that the $S$-orbit of $x_0$ does not depend on the choice of $x_0$:
\begin{claim} \label{claim:M}
All $S$-independent nodes in $O$ belong to the same $S$-orbit.
\end{claim}
\noindent
We call this $S$-orbit $M \subseteq O$ the \emph{starting} $S$-orbit inside $O$, and all its elements \emph{starting nodes}.

Intuitively, an obvious idea would be to declare $h_{n+1}(x_0) = f(x_0)$ and then to close
under atom automorphisms to lift the definition of $h_{n+1}$ to the full orbit $O$.
However, it might be the case that $y_0 = f(x_0) \in B$ is already contained in $\codom(h_n)$.
In such case we return back to $A$ using ${h_n}^{-1}$, which yields $x_1 := {h_n}^{-1}(y_0) \in \dom(h_n)$.
Continuing in this way, we define a sequence $x_0, y_0, x_1, y_1, x_2, \ldots$ of nodes, alternating between 
nodes of $\str A$ and nodes of $\str B$, by the following equalities:
\begin{align} \label{eq:alt}
h_n(x_{i+1}) \quad = \quad y_{i} \quad & = \quad f(x_i).
\end{align}
In words, $y_{i}$ is obtained from $x_i$ by applying $f$, and $x_{i+1}$ is obtained from $y_{i}$ by applying
$(h_n)^{-1}$. Clearly, the latter application is well-defined only when $y_{i} \in \codom(h_n)$. 
We thus stop generating the sequence as soon as $y_i \notin \codom(h_n)$.
We need however to prove that this will eventually happen, i.e., that the sequence is finite:
\begin{claim} \label{claim:length}
$y_l \notin \codom(h_n)$, for some $l \geq 0$.
\end{claim}
(In particular, when $n=0$ the claim necessarily holds for $l = 0$.)
The $S$-orbit $M'$ of $y_l$ will be called the starting $S$-orbit inside the $\emptyset$-orbit $O'$ of $y_l$ 
($\str A$ and $\str B$ are treated symmetrically here).
The number $l$ will be called the length of the starting orbit $M$, and $n+1$ will be called its order; we write
$\order{M} = n+1$, and $\length{M} = l$. 
By abuse of notation, we will also assign the same order and length 
to the $\emptyset$-orbit $O$, and to every element of $O$.

We are now ready to define the bijection $h_{n+1}$ with $\dom(h_{n+1}) = \dom(h_n) \cup O$.
It agrees with $h_n$ on $\dom(h_n)$; on the orbit $O$, we define $h_{n+1}$
by extending the mapping $x_0 \mapsto y_l$ to all of $O$.
We claim that there is a (unique) $\emptyset$-definable bijection between the $\emptyset$-orbit $O$ of $x_0$ and 
the $\emptyset$-orbit $O'$ of $y_l$ that maps $x_0$ to $y_l$:
\begin{claim} \label{claim:extends}
The set of pairs $\setof{(\pi x_0, \pi y_l)}{\pi \text{ is an atom automorphism}}$ is an $\emptyset$-definable bijection between $O$ and $O'$.
\end{claim}
\noindent
This completes the induction step of the definition of $h$.

Before proving that $h$ is an isomorphism, we formulate a concise equality describing $h$; the equality will be useful later.
The definition of $h_{n+1}$ in the induction step does not depend on the choice of the $S$-independent node $x_0 \in M$.
To see this, observe that Claim~\ref{claim:length} holds, with the same value $l$, for every other choice of $x_0 \in M$; indeed, both 
$f$ and $(h_n)^{-1}$ are $S$-definable, and hence by Lemma~\ref{lem:equivfunc1} they preserve 
the relation of belonging to the same $S$-orbit; 
thus, no matter which $S$-independent node $x_0 \in M$ is chosen, it will always belong to the same $S$-orbit (by Claim~\ref{claim:M}), hence the node
$y_l$ will always belong to the same $S$-orbit (and hence to the same $\emptyset$-orbit) inside $B$.
Thus by Lemma~\ref{lem:equivfunc1} and by the definition of $h$ we have,
%
for every starting node $x\in A$, the following equality:
\begin{align*} 
h(x) = \Big[ f \circ (h^{-1} \circ f)^l \Big] (x)
\end{align*}
where $l = \length{x}$; or equivalently (as $f$ and $h$ are bijections)
\begin{align} \label{eq:kolko}
(h^{-1} \circ f)^{l+1} (x) = x.
\end{align}
%
%
(In the special case of starting nodes $x$ of $\order{x} = 1$ (recall that 1 is the minimal possible order)  we have $\length{x} = 0$, 
and thus the above equality reduces to $h(x) = f(x)$. Thus $h$ agrees with $f$ on
the starting $S$-orbit of order 1.) 

Now we shall prove that $h$ is an isomorphism, i.e., that for every two nodes $x, x' \in A$,
\begin{align} \label{eq:iff}
(x, x') \in E \quad \text{ if, and only if } \quad (h(x), h(x')) \in F.
\end{align}
The proof is by induction on the orders of $x$ and $x'$.
Let $n = \order{x}$ and $n'  = \order{x'}$.
For the induction step, suppose that~\eqref{eq:iff} holds for all pairs $y, y'$ of starting nodes 
with $\langle m, m'\rangle = \langle\order{y}, \order{y'}\rangle$ pointwise strictly smaller than $\langle n, n' \rangle$.
Thus we assume that the claim holds for $y, y'$ whenever $m \leq n$, $m' \leq n'$, but either $m < n$ or $m' < n'$.
Let $l = \length{x}$ and $l' = \length{x'}$.

We first prove~\eqref{eq:iff} in the special case when both $x$ and $x'$ are starting nodes.
(Recall that in the special case $n = n' = 1$ we have $h(x) = f(x)$ and $h(x') = f(x')$, and the claim follows since $f$ is an isomorphism.)
Consider two sequences of nodes,
\begin{align*}
x_0, y_0, x_1, y_1, \ldots, 
\qquad \text{ and } \quad
x'_0, y'_0, x'_1, y'_1, \ldots
\end{align*}
where $x_0 = x$ and $x'_0 = x'$,
alternating between nodes of $\str A$ and $\str B$, determined by the equalities analogous to~\eqref{eq:alt}:
\begin{align*} 
h(x_{i+1}) \quad = \quad y_{i} \quad & = \quad f(x_i) &
h(x'_{i+1}) \quad = \quad y'_{i} \quad & = \quad f(x'_i).
\end{align*}
Analogously as before, the sequence is obtained by alternating applications of $f$ and $h^{-1}$.
In particular, $x_{i+1} = h^{-1}(f(x_i))$ and  $x'_{i+1} = h^{-1}(f(x'_i))$. 
Consider the smallest $k > 0$ such that $x_k = x$ and $x'_k = x'$.
Observe that equality~\eqref{eq:kolko} 
guarantees that such $k$ exists, for instance $k = (l+1) (l'+1)$ works.
Note that the equality~\eqref{eq:kolko} implies also that
\begin{align} \label{eq:h}
h(x) = y_{k-1} \qquad \text{ and } \qquad h(x') = y'_{k-1}.
\end{align}
Since $f$ is an isomorphism, for every $i \geq 0$ we have:
\begin{align} \label{eq:iffassone}
(x_i, x'_i) \in E \quad \text{ if, and only if } \quad (y_i, y'_i) \in F.
\end{align}
Using the inductive assumption we want to prove additionally 
\begin{align} \label{eq:iffasstwo}
(x_{i+1}, x'_{i+1}) \in E \quad \text{ if, and only if } \quad (y_i, y'_i) \in F,
\end{align}
for every $i$ such that $0 < i+1 < k$.
By the definition of $h$ we know that $\order{x_i} \leq n$ and $\order{x'_i} \leq n'$ for every $i \geq 0$.
Furthermore, observe that for every $0 < i < k$ we have $\order{x_i} < n$ or $\order{x'_i} < n'$; indeed, the equalities
$\order{x_i} = n$ and $\order{x'_i} = n'$, together with the bijectivity of $(h^{-1} \circ f)^i$, would imply $x_i = x$ and $x'_i = x'$.
In consequence, the induction assumption~\eqref{eq:iff} applies for every pair $(x_{i+1}, x'_{i+1})$ where $0 < i+1 < k$, which 
proves the equivalences~\eqref{eq:iffasstwo}.
Combining~\eqref{eq:iffassone} with~\eqref{eq:iffasstwo} we get
\begin{align*}
(x, x') \in E \quad \text{ if, and only if } \quad (y_{k-1}, y'_{k-1}) \in F
\end{align*}
which, by~\eqref{eq:h}, is exactly~\eqref{eq:iff}, as required.

Now we proceed to proving~\eqref{eq:iff} for arbitrary nodes $x, x'\in A$. To this end we will use the following fact,
easy to prove using $\omega$-categoricity and denseness of $\atoms$:
\begin{claim} \label{claim:starting}
For every two nodes $x, x'\in A$ there is an atom automorphism $\pi$ such that both $\pi x$ and $\pi x'$ are starting nodes.
\end{claim}
For any $\pi$ in Claim~\ref{claim:starting} we have the following sequence of equivalences:
\begin{align}  
\begin{aligned} \label{eq:finaliff}
(x, x') \in E  & \ \text{ iff } \\ 
(\pi x, \pi x') \in E  & \ \text{ iff } \\  
(h(\pi x), h(\pi x')) \in F  & \ \text{ iff } \\
(\pi h(x), \pi h(x')) \in F  & \ \text{ iff } \\ 
(h(x), h(x')) \in F &.
\end{aligned}
\end{align}
The first and the last equivalence hold as the relations $E$ and $F$, being $\emptyset$-definable, are closed under atom automorphisms
 (cf.~Lemma~\ref{lem:closed});
the second equivalence has been treated previously as both $\pi x$ an $\pi x'$ are starting nodes;
and the third equivalence holds as the function $h$ is $\emptyset$-definable (cf.~Lemma~\ref{lem:equivfunc1}).
Theorem~\ref{thm:elim} is thus proved, once we prove the yet unproven Claims~\ref{claim:orbit}--\ref{claim:starting}.
%

\begin{proof}[Proof of Claim~\ref{claim:orbit}]
Consider an $\emptyset$-orbit $O \subseteq A\cup B$. We will use the embedding $H : \atoms \to \atoms$.
%
Take any node $x\in O$ and consider the restriction of $H$ to the finite set $\supp{x} \subseteq \atoms$.
As $H$ is an embedding, the restriction is a finite partial isomorphism which extends, by $\omega$-categoricity of $\atoms$
(c.f.~Lemma~\ref{lem:extend}), 
to an atom automorphism, say $\pi$.
Then $\pi x$ is an $S$-independent node in $O$, as $\supp{\pi x} = \pi \supp{x} \subseteq H(\atoms)$
(cf.~Lemma~\ref{lem:supp}).
\end{proof}

\begin{proof}[Proof of Claim~\ref{claim:M}]
Let $x, x' \in O$ be two $S$-independent nodes, thus $\supp{x} \cup \supp{x'} \subseteq H(\atoms)$.
Take an atom automorphism $\pi$ such that $\pi x = x'$.
By Lemma~\ref{lem:supp} we have $\pi \supp{x} = \supp{x'}$.
The restriction of $\pi$ to $\supp{x}$ is a finite partial isomorphism, and hence it extends to an automorphism of $H(\atoms)$, which in turn
extends to an atom $S$-automorphism $\tau$, by denseness. Thus $\tau x = x'$
(we use Lemma~\ref{lem:eqaction} here) and hence $x$ and $x'$ are in the same $S$-orbit.
\end{proof}

\begin{proof}[Proof of Claim~\ref{claim:length}]
It is enough to prove that the nodes $y_0, y_1, \ldots$ belong to different $S$-orbits inside $B$ (as $B$ has finitely many $S$-orbits, and so does 
$\codom(h_n)$, and the claim follows).
With this aim consider the sequence of $S$-orbits:
\[
\orbit{S}{x_0}, \orbit{S}{y_0}, \orbit{S}{x_1}, \ldots 
\]
and suppose, towards contradiction, that some $S$-orbit $O$ in $B$ repeats in the sequence, say
$\orbit{S}{y_k} = \orbit{S}{y_j}$ for some indices $j < k$.
Recall that $A$ and $B$ are disjoint, and that $y_i = f(x_i)$ and $x_{i+1} = h_n^{-1}(y_i)$, for every $i$.
Since $f$ is $S$-definable, and $h_n$, being $\emptyset$-definable, is also $S$-definable, by Lemma~\ref{lem:equivfunc1} 
they both (preserve and) reflect the relation of belonging to the same $S$-orbit; therefore 
$\orbit{S}{y_{k-\alpha}} = \orbit{S}{y_{j-\alpha}}$ for $\alpha = 1, \ldots, j$, and in consequence
the first orbit $O$ necessarily repeats in the sequence:
\begin{align} \label{eq:imposs}
\orbit{S}{x_0} \ = \ \orbit{S}{x_i}
\end{align}
for $i = k-j > 0$. Recall that $x_i = (h_n^{-1}\circ f)^i$ to observe that the equality~\eqref{eq:imposs} is impossible:
$x_i \in \dom(h_n)$ while $x_0 \notin \dom(h_n)$, and $\dom(h_n)$, being the union of $\emptyset$-orbits,
is also the union of $S$-orbits.
\end{proof}

\begin{proof}[Proof of Claim~\ref{claim:extends}]
Relying on Lemma~\ref{lem:equivfunc2} it is enough to demonstrate that the set of pairs 
\[
\setof{(\pi x_0, \pi y_l)}{\pi \text{ is an atom automorphism}}
\]
is (the graph of) a bijection. In other words, it is enough to prove that for every two atom automorphism $\pi, \pi'$, 
the equality $\pi x_0 = \pi' x_0$ holds if, and only if the equality $\pi y_l = \pi' y_l$ holds.
This is equivalent to the following condition:
\begin{align} \label{eq:implies}
\text{for every atom automorphism } \pi, \ \pi x_0 = x_0 \text{ iff } \pi y_l = y_l.
\end{align}
%
%
%
Recall that $y_l = g(x_0)$, where $g = \Big[ f \circ (h_n^{-1} \circ f)^l \Big]$; as both $f$ and $h_n$ are $S$-definable (partial) bijections,
their composition $g$ is also an $S$-definable bijection, and hence the condition~\eqref{eq:implies} holds for all atom $S$-automorphisms $\pi$.
Our aim is to prove~\eqref{eq:implies} for all atom automorphisms, using $S$-independence of $x_0$.
%
%

Define the \emph{$S$-support} as $\supparg{S}{x} \ \eqdef \ \supp{x} \setminus S$.
Let $U = \supp{x_0}$. We claim that $\supp{y_l} = U$ as well, i.e., $y_l$ is $S$-independent too. 
To demonstrate this, we first observe using Lemma~\ref{lem:supfun} that the mapping $g$, being $S$-definable, can only decrease the $S$-support:
$
\supparg{S}{x_0} \supseteq \supparg{S}{y_l}.
$
But $g^{-1}$, being also $S$-definable, has the same property and hence
$
\supparg{S}{x_0} = \supparg{S}{y_l}.
$
Furthermore, $x_0$ is $S$-independent and thus satisfies $\supparg{S}{x_0} = \supp{x_0}$, which yields 
$U = \supparg{S}{y_l} \subseteq \supp{y_l}$. Finally, the dimension of $y_l$ is at most equal to the dimension of $x_0$
(because the node $x_0$ has been chosen as one with the maximal dimension), i.e., to the cardinality of $U$, 
therefore we deduce the equality $U = \supp{y_l}$.

We are now ready to prove~\eqref{eq:implies}; we focus on the left-to-right implication, as the other one is proved similarly.
Consider any atom automorphism $\pi$ satisfying $\pi x_0 = x_0$. 
By Lemma~\ref{lem:supp} we know that $\pi$ preserves the set $U$, i.e., $\pi U = U$.
We claim that some atom $S$-automorphism $\pi'$ coincides with $\pi$ on $U$.
Indeed, the restriction of $\pi$ to $U$ extends, as a finite partial isomorphism inside $H(\atoms)$, to an automorphism of $H(\atoms)$;
and the latter extends to an atom $S$-automorphism $\pi'$, by denseness of $\atoms$.
%
%
As $\pi$ and $\pi'$ coincide on $U = \supp{x_0}$,
by Lemma~\ref{lem:eqaction} we have $\pi' x_0 = \pi x_0 = x_0$. We can apply~\eqref{eq:implies} to the atom $S$-automorphism
$\pi'$, thus obtaining $\pi' y_l = y_l$.
Finally, again by Lemma~\ref{lem:eqaction} applied to $\pi$ and $\pi'$, coinciding on $U = \supp{y_l}$, we deduce
$\pi y_l = \pi' y_l$. This entails $\pi y_l = y_l$ as required.
\end{proof}

\begin{proof}[Proof of Claim~\ref{claim:starting}]
Similarly as in the proof of Claim~\ref{claim:orbit}, we use the embedding $H : \atoms \to \atoms$.
The restriction of $H$ to the finite set $\supp{x} \cup \supp{x'} \subseteq \atoms$ is necessarily a finite partial isomorphism which extends,
by $\omega$-categoricity of $\atoms$, to an atom automorphism $\pi$, such that the nodes
$\pi x$ and $\pi x'$ are both $S$-independent and hence, by Claim~\ref{claim:M}, also starting.
\end{proof}

\para{The general case} 
The additional assumptions imposed in the proof are not essential, and the proof easily adapts to the general case.
To get rid of the assumption $T = \emptyset$, one just needs to replace the atoms by the structure $\atoms'$ obtained from $\atoms$ by introducing constant symbols for each element of $T$. Then, every set $x$ which is $T$-definable in $\atoms$ becomes $\emptyset$-definable in $\atoms'$. Moreover, $\atoms'$ remains effectively $\omega$-categorical, dense, and still has least supports.

One also easily gets rid of the assumption that the signature has just one binary symbol.
First, to deal with (possibly infinitely) many symbols, in the inductive argument towards $h$ being an isomorphism one treats each symbol separately.
In case symbols of arity other than 2, say $r$, the induction is with respect to the $r$-tuples of orders, instead of pairs thereof.
The inductive argument itself, as well as Claim~\ref{claim:starting}, adapt easily to $r$-tuples.
Finally, in case of a (possibly infinite) $T$-definable signature, one needs to modify the sequence of equalities~\eqref{eq:finaliff} appropriately, in order to
take into account the action of atom automorphisms on signature symbols. For a signature symbol $\sigma$, denote by $\sigma^{\str A}$, 
$\sigma^{\str B}$ the interpretation
of $\sigma$ in $\str A$, $\str B$, respectively. Then the sequence of equalities~\eqref{eq:finaliff} is adapted as follows:
\begin{align*}  
(x_1, \ldots, x_r) \in \sigma^{\str A}  & \ \text{ iff } \\ 
(\pi x_1, \ldots, \pi x_r) \in \pi (\sigma^{\str A})  & \ \text{ iff } \\  
(\pi x_1, \ldots, \pi x_r) \in (\pi \sigma)^{\str A}  & \ \text{ iff } \\  
(h(\pi x_1),\ldots, h(\pi x_r)) \in (\pi \sigma)^{\str B}  & \ \text{ iff } \\
(h(\pi x_1),\ldots, h(\pi x_r)) \in \pi (\sigma^{\str B})  & \ \text{ iff } \\
(\pi h(x_1), \ldots, \pi h(x_r)) \in \pi (\sigma^{\str B})  & \ \text{ iff } \\ 
(h(x_1), \ldots, h(x_r)) \in \sigma^{\str B} &.
\end{align*}
We consider here the action $\pi \sigma$ on the signature symbol $\sigma$, as well as the action on the interpretation of the signature symbol,
$\pi (\sigma^{\str A})$ or $\pi (\sigma^{\str B})$. In particular, $(\pi \sigma)^{\str A}$ denotes the interpretation of the signature symbol $\pi \sigma$ in $\str A$.
As $\pi$ is a $T$-automorphism, the first, the second, the fourth and the last equivalence follow by $T$-definability of $\str A$ and $\str B$
(cf.~Lemma~\ref{lem:closed});
the third equivalence is proved previously, as $\pi x_1, \ldots, \pi x_r$ are starting nodes;
and the fifth equivalence holds since the function $h$ is $T$-definable (cf.~Lemma~\ref{lem:equivfunc1}).